\documentclass[twocolumn,preprintnumbers,amsmath,amssymb,a4paper,aps,pra,10pt]{revtex4-1}
\usepackage{amssymb}
\usepackage{amsfonts}
\usepackage{amsmath}
\usepackage{amsthm}
\usepackage{graphicx}
\usepackage{braket}
\usepackage{bm}
\usepackage{epstopdf}

\newcommand{\ketbra}[2]{\ensuremath{\left| #1 \right\rangle \left\langle #2 \right|}}

\newtheorem{thm}{Theorem}

\newtheorem{lemma}{Lemma}
\newtheorem{mydef}{Definition}

\begin{document}

\title{What would it have looked like if it looked like I were in a superposition?}
\author{Daniel Filan, Joseph J. Hope}
\affiliation{Department of Quantum Science, Research School of Physics and Engineering, Australian National University, Canberra ACT 0200, Australia}
\date{\today}

\begin{abstract}
In this paper we address the question of whether it is possible to obtain evidence that we are in a superposition of different `worlds', as suggested by the relative state interpretation of quantum mechanics. We find that it is impossible to find definitive proof, and that if one wishes to retain reliable memories of which `world' one was in, no evidence at all can be found.  We then show that even for completely linear quantum state evolution, there is a test that can be done to tell if you can be placed in a superposition. 
\end{abstract}

\maketitle

\section{Introduction}
Quantum mechanical objects can be placed in a superposition of states.  
How large these superpositions can be is very much an open question, particularly given the possibility of the existence of intrinsic universal decoherence processes that resolve the measurement problem, or limitations suggested by some variations of quantum gravity \cite{Arndt2014}.  
It is also a key question to answer given that macroscopic superpositions are critical to proposed new quantum technologies \cite{Ladd2010}.  
The relative state interpretation (RSI) of quantum mechanics, as originally devised by Everett, admits no upper limit, and strongly suggests that we are all in macroscopic superposition states.  It is often suggested that in a universe that evolves linearly, it is impossible to test whether we are in such a superposition, and therefore that the truth of this proposition is not a question that can be addressed by scientific method.  This paper will show that it is formally impossible to determine whether we are currently in a superposition, but that it is possible, at least in principle, to determine experimentally whether we can be placed in such a superposition.

The RSI assumes that there is no physical process of wavefunction collapse upon a measurement, and that the quantum state of the universe only ever undergoes unitary evolution \cite{Everett57,Wallace2012}.  With no wavefunction collapse process, there is nothing preventing arbitrarily large systems becoming entangled, including human experimenters.  This possibility, which historically has been seen as very confronting \cite{Byrne2013}, is precisely how the RSI explains the fact that our everyday experiences seem clearly single-valued.  
If we assume that our experiences are a direct function of the state of our physical brains, then if a brain goes into a superposition while entangled with some quantum system and a measuring apparatus, then the resulting state is best described as a superposition of ``quasi-classical" experiences, in which each brain state is entangled with a single measurement result. 

The RSI requires the fewest axioms to describe, but requires considerable development to explain the quasi-classical experience.  Unlike most interpretations of quantum mechanics, it requires no definition of measurement to be complete.  Unlike the Copenhagen-style interpretations, it has no divide between the quantum and classical world, and is a complete theory that makes firm predictions about mesoscopic systems.  With no non-local evolution at all, it is entirely compatible with all relativistic theories.  Any theory that makes predictions \textit{exactly} as though there were nothing but unitary evolution, we will regard as in an equivalence class with the RSI theory.  Fully defined quantum theories that enforce single-valuedness typically introduce a physical wavefunction collapse process \cite{GRW,tumulka2006relativistic}.  Any wavefunction collapse process will reduce the coherence of subsequent experiments, and these theories choose their parameters such that this collapse is undetectably slow on the scale of current quantum experiments, and undetectably fast for macroscopic superpositions.  

It is well-known how to determine if an external system is in a superposition of distinct states. For example, Eibenberger \textit{et al.}~\cite{eibenberger2013} demonstrate that molecules exceeding a mass of 10,000 amu can be put into a superposition of different spatial positions.  Scaling this process up will either be impossible due to physical decoherence processes such as wavefunction collapse, or else simply a function of technology.  In principle, doing the same thing to a cat or a human is a question of scale.  If it were possible to demonstrate that another human could be in a superposition, then we might be inclined to believe that something like the RSI is correct.  What is perhaps less clear is how to determine if oneself is in a superposition of distinct quasi-classical states, or if such a thing is possible. Some previous proposed tests of the RSI, such as \cite{Deutsch1986} and \cite{Vaidman1998}, have not addressed this aspect.  Other proposed tests (e.g. \cite[Chapter 8]{albert1992quantum}) purport to cause someone to become aware that they are in a superposition of different mental states, but these treatments do not correctly account for the non-orthogonality of states in superpositions and states not in a superposition.  This paper will examine tests that demand contrasting results for states in a superposition and those not in a superposition.  We will show that it is not possible to determine experimentally whether one's own mind is currently in a superposition, though it is possible to show that the process can occur.

\section{Assumptions regarding experience and experiments}

This paper examines the potential experiments that can be performed in a universe that is completely determined by a state vector living in a complex Hilbert space with unitary time evolution. When we say that a subspace of this state is an ``experimenter'' who is capable of believing things, we assume that those beliefs are purely the function of the state of that subspace.  In practice, we imagine this subspace to be something like the experimenter's brain, though the precise details of the scope of these structures is not important, providing it can be encoded within the larger Hilbert space.

One detail that is required for our proofs is that any two states of the universe where the experimenter believes contradictory things must be orthogonal.  Under any standard interpretation of quantum mechanics, if two states are not orthogonal, then it is not possible to distinguish between the two with perfect fidelity, which we assume must be possible for contradictory beliefs.

To illustrate how these non-probabilistic assumptions will be used, we give a definition and then prove a simple lemma:

\begin{mydef}
  Say that states of the universe $\ket{1}$ and $\ket{2}$ are \emph{distinguishable by the experimenter} (or simply \emph{distinguishable}) if there is some unitary transformation $\hat{U}$ such that the experimenter believes something different in $\hat{U} \ket{1}$ and $\hat{U} \ket{2}$.
\end{mydef}

\begin{lemma}
Distinguishable states are orthogonal.
\end{lemma}

\begin{proof}
  $\braket{1|2} = \braket{1 | \hat{U}^\dagger \hat{U} | 2}$. By hypothesis, $\hat{U} \ket{1}$ and $\hat{U} \ket{2}$ contain experimenters that believe different things, so $\braket{1 | \hat{U}^\dagger \hat{U} | 2} = 0$. Therefore, $\braket{1|2} = 0$.
\end{proof}

What this lemma tells us is that if two states are distinguishable, there must be some unitary transformation that takes them to an orthogonal state. Since unitary transformations preserve inner products, the states must have been orthogonal to begin with.  This definition and lemma could be easily operated in reverse, and we could motivate that contradictory beliefs must be orthogonal by noting that they can be created by performing an experiment that distinguishes between two completely distinct (orthogonal) initial states.


\section{Assumptions regarding probability}
Later we shall investigate experimental tests of whether one is in a superposition that do not give definitive answers, but still produce non-identical quantum states consisting of superpositions of `Yes' and `No' answers.  This is interpreted as a probabilistic outcome in all interpretations of quantum mechanics, although there is much debate over the validity of doing this in the RSI (see for example \cite{AlbertProb2010}).While it can be shown that the Born rule is the unique way to assign valid probabilities \cite{Everett57}, and can also be justified through considerations of decision theory \cite{Wallace2012} or of self-locating uncertainty \cite{carroll2014many}, for our purposes we shall simply assume that one can sensibly talk about probabilities.  That is to say, if the experimenter measures the spin of an electron in the state $\alpha \ket{\uparrow} + \beta \ket{\downarrow}$, we will make statements such as ``The probability of the experimenter measuring $\uparrow$ is $|\alpha|^2$", while also maintaining that both outcomes of the measurement actually happen, and that the whole experiment proceeds in an entirely deterministic manner.

The question that this paper asks is: ``if the RSI were correct, giving a satisfactory account of probability, and we were in a large-scale superposition of so-called `worlds', would we be able to test this?", without considering the plausibility or lack thereof of the premise.

\section{Impossible Types of Superposition Tests}

The first type of superposition test that we might want is a definitive one: that is, one that will definitely give a ``Yes" answer if the universe started out in a superposition, and definitely give a ``No" answer if the universe did not. This motivates our next definition:

\begin{mydef}
  Call a unitary operator $\hat{T}_D$ a \emph{definitive superposition test} if for each $\ket{i}$,
    \begin{align*}
      \hat{T}_D \ket{i} &= \sum_j \beta_{ij} \ket{N_{ij}} \\
      \intertext{and}
      \hat{T}_D \left( \sum_i \alpha_i \ket{i} \right) &= \sum_k \gamma_k \ket{Y_k}
    \end{align*}
for some set of $\alpha_i$, where $\| \sum_i \alpha_i \ket{i} \| = 1$, and such that each pair of $\ket{N_{ij}}$ and $\ket{Y_k}$ are distinguishable.
\end{mydef}
The states $\ket{N_{ij}}$ are negative results (`No, I wasn't in a superposition') and the $\ket{Y_k}$ are positive results (`Yes, I was in a superposition').

\begin{thm}
  No definitive superposition test exists.
\end{thm}

\begin{proof}
 Suppose that $\hat{T}_D$ were a definitive superposition test. Then, we would have
   \begin{align*}
      \hat{T}_D \left( \sum_i \alpha_i \ket{i} \right) &= \sum_k \gamma_k \ket{Y_k} \\
      &= \sum_i \alpha_i \hat{T}_D \ket{i} \\
      &= \sum_{ij} \alpha_i \beta_{ij} \ket{N_{ij}}
    \end{align*}
  However, since each $\ket{N_{ij}}$ and $\ket{Y_k}$ is distinguishable, we have
    \begin{align*}
      \left\| \sum_k \gamma_k \ket{Y_k} \right\|^2 &= \left( \sum_k \gamma_k^* \bra{Y_k} \right) \left( \sum_{ij} \alpha_i \beta_{ij} \ket{N_{ij}} \right) \\
      &= \sum_{ijk} \gamma_k^* \alpha_i \beta_{ij} \braket{Y_k | N_{ij}} \\
      &= 0
    \end{align*}
  However, $\sum_i \alpha_i \ket{i}$ was normalised and $\hat{T}_D$ was unitary by hypothesis. Therefore, we have a contradiction.
\end{proof}

As perhaps expected, the linearity and unitarity of the evolution makes it impossible to prove that we are in a superposition.  Upon learning this, we might want to settle for a test that had inconclusive results, but at least some definite `Yes' results.  In this way, if we ever experience a definite `Yes', then we have learned that we were in a superposition.  This leads us to the next definition:

\begin{mydef}
  Call a unitary operator $\hat{T}_{PD}$ a \emph{partially-definitive superposition test} if
    \begin{align*}
       \hat{T}_{PD} \ket{i} &= \sum_j \beta_{ij} \ket{N_{ij}}, \text{  and}\\
       \hat{T}_{PD} \left( \sum_i \alpha_i \ket{i} \right) &= \sum_{ij} \gamma_{ij} \ket{N_{ij}} + \sum_k \zeta_k \ket{Y_k}
    \end{align*}
  for some set of $\alpha_i$, where $\| \sum_i \alpha_i \ket{i}\| = 1$, $\sum_k \zeta_k \ket{Y_k}$ is non-zero, and such that each pair of $\ket{N_{ij}}$ and $\ket{Y_k}$ are distinguishable.
\end{mydef}

We note that this time, the states $\ket{N_{ij}}$ stand for `No result', as we have not definitively shown that we are not in a superposition when we are in that state.  This is still a contradictory belief to all states $\ket{Y_k}$, as they have clearly resolved disagreement over whether the test has returned a definitive result.

\begin{thm}
  No partially-definitive superposition test exists.
\end{thm}

\begin{proof}
  Suppose that $\hat{T}_{PD}$ were a partially-definitive superposition test. Then, we would have
    \begin{align*}
      \hat{T}_{PD} \left( \sum_i \alpha_i \ket{i} \right) &= \sum_{ij} \gamma_{ij} \ket{N_{ij}} + \sum_k \zeta_k \ket{Y_k} \\
      &= \sum_i \alpha_i \hat{T}_{PD} \ket{i} \\
      &= \sum_{ij} \alpha_i \beta_{ij} \ket{N_{ij}} \\
      \sum_k \zeta_k \ket{Y_k} &= \sum_{ij} (\alpha_i \beta_{ij} - \gamma_{ij} ) \ket{N_{ij}}
    \end{align*}
  However, since each $\ket{N_{ij}}$ and $\ket{Y_k}$ is distinguishable,
    \begin{align*}
      \left\| \sum_k \zeta_k \ket{Y_k} \right\|^2 &= \left( \sum_k \zeta_k^* \bra{Y_k} \right) \left( \sum_{ij} (\alpha_i \beta_{ij} - \gamma_{ij} ) \ket{N_{ij}} \right) \\
      &= \sum_{ijk} \zeta_k^* (\alpha_i \beta_{ij} - \gamma_{ij} ) \braket{Y_k | N_{ij}} \\
      &= 0
    \end{align*}
  Therefore, $\sum_k \zeta_k \ket{Y_k}$ is the zero vector, a contradiction.
\end{proof}

We next wonder if we could construct some probabilistic test of whether we were in a superposition. We would like this test to allow reliable `memories' of the distinguishable branches that the experimenter was in immediately before the test.  Without this ability for the test to preserve memory, it is impossible to answer the question: ``Am I, right now, in a superposition, even though my world currently seems quasi-classical?", as those quasi-classical experiences must be destroyed by the test itself.

This leads us to the formal definition:
 
 \begin{mydef}
   Call a unitary operator $\hat{T}_{BD}$ a \emph{branch-discriminating probabilistic superposition test} if
     \begin{align*}
       \hat{T}_{BD} \ket{i} &= \sum_j (\beta_{ij} \ket{N_{ij}} + \gamma_{ij} \ket{Y_{ij}} ), \text{   and}\\
       \hat{T}_{BD} \left( \sum_i \alpha_i \ket{i} \right) &= \sum_{ij} (\zeta_{ij} \ket{N_{ij}} + \eta_{ij} \ket{Y_{ij}} )
     \end{align*}
for some set of $\alpha_i$ such that $\| \sum_i \alpha_i \ket{i} \| = 1$, each pair of $\ket{N_{ij}}$ and $\ket{Y_{i'j'}}$ is distinguishable, $\ket{Y_{ij}}$ and $\ket{Y_{i'j'}}$ are distinguishable for $i \neq i'$, and there exists some $i$ such that
  \begin{align*}
   \left\| \sum_j \gamma_{ij} \ket{Y_{ij}} \right\|^2 &< \left\| \sum_j \eta_{ij} \ket{Y_{ij}} \right\|^2
  \end{align*}
 \end{mydef}

The condition that $\ket{Y_{ij}}$ and $\ket{Y_{i'j'}}$ are distinguishable encodes our desire to be able to discriminate between before-test branches after the test, and the condition that $\| \sum_j \gamma_{ij} \ket{Y_{ij}} \|^2 < \| \sum_j \eta_{ij} \ket{Y_{ij}} \|^2$ ensures that for at least one of the branches, there is some set of `Yes' outcomes that, taken together, are more likely if you are in a superposition than if you were in a single-valued state.

\begin{thm}
  No branch-discriminating probabilistic superposition test exists.
\end{thm}

\begin{proof}
  Without loss of generality, we may assume that $\braket{Y_{ij} | Y_{ik}} = \delta_{jk}$ -- that is, that the `Yes' states for each branch are orthogonal to each other. If not, we can find an orthonormal basis for the span of the $\ket{Y_{ij}}$ (with fixed $i$ and variable $j$) and write the superpositions in terms of these bases. This preserves squared norm and orthogonality of $\ket{Y_{ij}}$ and $\ket{Y_{i'j'}}$, and means that the probabilistic condition may be restated as
    \begin{align*}
      \sum_j |\gamma_{ij}|^2 < \sum_j |\eta_{ij}|^2
    \end{align*}

Now, by linearity,
  \begin{align*}
    \sum_{ij} (\zeta_{ij} \ket{N_{ij}} + \eta_{ij} \ket{Y_{ij}} ) &= \sum_{ij} ( \alpha_i \beta_{ij} \ket{N_{ij}} + \alpha_i \gamma_{ij} \ket{Y_{ij}} )
  \end{align*}
By taking the inner product of both sides with each of the $\ket{Y_{ij}}$, we see that $\eta_{ij} = \alpha_i \gamma_{ij}$, and therefore that
   \begin{align*}
    \sum_j |\eta_{ij}|^2 &= \sum_j |\alpha_i|^2 |\gamma_{ij}|^2 \\
    &\leq \sum_j |\gamma_{ij}|^2
  \end{align*}
This contradicts our probabilistic condition.
\end{proof}

\section{Allowable Tests}

The previous theorems do not rule out working superposition tests. Theorems 1 and 2 show that such a test must be probabilistic, and theorem 3 shows that no probabilistic test is possible when the memories of being in each part of the superposition are retained.
In other words, in order to produce the necessary interference that can determine the existence of a superposition, multiple branches must be coupled to exactly the same final state, including states of belief of the experimenter.  

Tests of this kind are schematically rather simple.
Suppose that the initial state of the universe is either $\ket{1}$, $\ket{2}$, or $(\ket{1} + \ket{2})/\sqrt{2}$. Then, we could have a test $\hat{T}$ defined by $\hat{T} \ket{1} = (\ket{N} + \ket{Y}) / \sqrt{2}$ and $\hat{T} \ket{2} = (- \ket{N} + \ket{Y}) / \sqrt{2}$. Then, by linearity, $\hat{T} (\ket{1} + \ket{2} ) / \sqrt{2} = \ket{Y}$, i.e.\ the outcomes interfere, so finding oneself in $\ket{Y}$ is evidence that one was originally in a superposition. 

To picture this test, imagine an experimenter who wants to perform an experiment that would distinguish between the RSI and any objective collapse interpretation. They enter an isolated machine with a piece of blank paper and an electron in the spin state $(\ket{\uparrow} + \ket{\downarrow})/\sqrt{2}$.  The machine can apply arbitrary unitary transformations to its contents. Inside, the experimenter measures the electron in the $\{ \ket{\uparrow}, \ket{\downarrow}\}$ basis. If the RSI is correct, the resulting state is $( \ket{\uparrow}\ket{E[\uparrow]} + \ket{\downarrow} \ket{E[\downarrow]}) / \sqrt{2}$, where $\ket{E[\uparrow]}$ is the state of the experimenter seeing a spin-up electron, and $\ket{E[\downarrow]}$ is the state of the experimenter seeing an spin-down electron. 
If an objective collapse interpretation is correct, the resulting state is either $\ket{\uparrow} \ket{E[\uparrow]}$ or $\ket{\downarrow} \ket{E[\downarrow]}$ with equal probability. The enclosing machine then applies any transformation to the room satisfying  
  \begin{align*}
    \ket{\uparrow} \ket{E[\uparrow]} &\rightarrow (\ket{N} + \ket{Y}) / \sqrt{2} \\
    \ket{\downarrow} \ket{E[\downarrow]} &\rightarrow (- \ket{N} + \ket{Y}) / \sqrt{2}
  \end{align*}
  
The experimenter's memories are uncorrelated to the measurement of the electron, but the piece of paper says ``No" if they happen to be in the state $\ket{N}$, and ``Yes" if they happen to be in the state $\ket{Y}$. They then walk outside and put their piece of paper on a table.  The experiment can then be repeated an arbitrary number of times.   

Over time, if there is a roughly equal number of ``No" and ``Yes" results, then the experimenter can become confident that the experiment is showing that they are not existing as a coherent superposition of states within the device.  Alternatively, if all of the results say ``Yes", then the experimenter has evidence of having been personally in a superposition, thus giving an arbitrary degree of confidence in the RSI.

We also note that this test relies crucially on both the `memory loss' experienced by the experimenter, and the knowledge of the phase of the initial superposition.  The full quantum state of the experimenter, including their memories, is being generated by the machine.  This means that it is possible for them to have any memories at all, but we have shown that they must be identical across multiple branches of the superposition, and therefore cannot be correlated with the actual relevant measurement results.  The previous test proposed by Deutsch \cite{Deutsch1986}shows that it is theoretically possible to show that a separate sentient being can be placed in a superposition, our allowable test shows that it is also possible for an experimenter to show that they themselves can be placed in a superposition.

By noting the equality of the density matrices, it is trivial to show that no experiment can distinguish between a complete lack of knowledge of the phase of the initial superposition and a stochastic mixture of single-valued states:

\begin{widetext}
    \begin{align*}
      \hat{\rho} &= \int_0^{2\pi} \frac{d\phi_1}{2\pi} \int_0^{2\pi} \frac{d\phi_2}{2\pi} \dotsi \int_0^{2\pi} \frac{d\phi_N}{2\pi} \left( \sum_j e^{i\phi_j} \alpha_j \ket{j} \right) \left( \sum_{j'} e^{-i\phi_{j'}} \alpha_{j'}^* \bra{j'} \right) \\
      &= \sum_j \alpha_j \alpha_j^* \ketbra{j}{j}
    \end{align*}
\end{widetext}
This emphasises that it is never possible to prove experimentally that a system is in an arbitrary superposition, only in a specific superposition.  So it is only possible to aspire to prove that we can be put into a specific superposition of quasi-classical states.

\section{Discussion and Conclusions}
It is clear that experiments that show increasingly large coherences can narrow the parameter regime in which spontaneous collapse theories might exist, but there is an enormous gap between current experiments and coherence experiments on truly macroscopic objects.  To make things even more challenging, while the majority of spontaneous collapse theories depend only on system size to define a decoherence rate, there are some theories that postulate the existence of genuinely classical objects and/or fields, although there are strong constraints on such theories \cite{Kapustin2013}.
Therefore, to avoid all loopholes, it would be eventually be necessary to include an actual sentient being in the interference experiment.  This paper characterises the limits of extending that experiment to its logical conclusion - personal experience.

We have characterised the limitations on potential tests of the RSI, showing that they cannot be definitive, and probabilistic tests must be unable to retain any knowledge of the contradictory beliefs held during a macroscopic superposition.  We show that it is nevertheless theoretically possible to demonstrate that oneself can be placed in a superposition of distinct quasi-classical states to any degree of certainty (below absolute) that is desired.

\bibliography{mwiPaper}
%
%
%

\end{document}